\documentclass{article}

\usepackage{arxiv}

\usepackage[utf8]{inputenc}
\usepackage[T1]{fontenc}
\usepackage{microtype}

\usepackage{amsmath,amssymb}
\usepackage{amsthm}
\theoremstyle{plain}
\newtheorem{theorem}{Theorem}

\newtheorem{corollary}{Corollary}
\theoremstyle{definition}

\usepackage{graphicx}
\graphicspath{{figures/}}
\usepackage{booktabs}
\usepackage{siunitx}
\usepackage{xcolor}

\usepackage[numbers,sort&compress]{natbib}
\usepackage{hyperref}
\hypersetup{colorlinks=true, linkcolor=blue!55!black,
            citecolor=blue!55!black, urlcolor=blue!55!black}
\usepackage[capitalize,noabbrev]{cleveref}

\newcommand{\tool}{\textsc{figverify}\xspace}

\providecommand{\chem}[1]{\ensuremath{\mathrm{#1}}}

\newcommand{\float}[1]{\texttt{Float#1}}

\usepackage{xspace}

\title{Automated High-Precision Extraction and Forensic Verification\\
       of Data-Bearing Vector Figures}

\author{%
  Bowen Sun\thanks{Corresponding author.} \\
  Johns Hopkins University \\
  \texttt{bsun39@jh.edu}
  \And
  Chaowei Xiao \\
  Johns Hopkins University \\
  \texttt{chaoweixiao@jhu.edu}
}

\date{Code and data:\quad\href{https://github.com/Bowen-Sun-0728/figverify}%
      {\texttt{github.com/Bowen-Sun-0728/figverify}}}

\renewcommand{\headeright}{Preprint}
\renewcommand{\undertitle}{Preprint}
\renewcommand{\shorttitle}{Automated Extraction \& Forensic Verification of Vector Figures}

\begin{document}
\maketitle

\begin{abstract}
The quantitative record of science and engineering is increasingly carried by
figures rather than text or tables, and a reader who needs the underlying numbers
must usually re-digitize them by hand: slowly, imprecisely, and with no way to prove
the result is faithful. Yet when a figure is stored as vector graphics, its data are
not approximated by the picture but encoded in it: the renderer writes each marker
and vertex at a printed precision that, for the dominant scientific toolchain,
exceeds the data's own. We turn this into three contributions, one per shortcoming of
hand digitization. First, a precision theory bounding how accurately data can be
recovered for a given renderer and export format: bit-exact \float{32} for
\texttt{matplotlib} markers, and a calibration-limited three to four significant
figures end to end. Second, an automatic extractor that decodes a figure in one pass
with no human in the loop, in place of the slow point-by-point tracing a digitizer
demands. Third, a verification theory: recovery is injective except on a
characterized, vanishingly small interval near zero; accidental agreement between
unrelated data is astronomically unlikely; and a re-rendering certificate binds the
recovered values to the markers, lines, and ticks the figure draws, not its text,
making a result non-repudiable. With no ground truth used during recovery, decoded
figures match external archives (Planck~2018 to $\sim\!10^{-9}$; the Keeling
\chem{CO_2} record to $\sim\!5\times10^{-4}$), and one decoded figure independently
reproduces a correction to the Chinchilla scaling-law confidence interval. We map the
achievable precision across common renderers and their PDF, SVG, and EPS formats.
What we deliver is certified data; the scientific significance of any particular
dataset lies outside this paper's scope, and recovered values are candidates for
human review, never accusations.
\end{abstract}


\section{Introduction}
\label{sec:intro}

A growing share of scientific and engineering results now reaches the
reader as a picture. A scaling law is a line through a cloud of points; a measured
spectrum is a curve with error bars; a device characteristic is a family of curves
on a datasheet. The numbers behind these pictures are often not printed in the
surrounding text, nor supplied as a table, nor deposited as data: the figure itself
is the record. A reader who wants those numbers, whether to reproduce an analysis, combine
results across papers, or check a headline claim, is forced to recover them from
the figure.

Today that recovery is done by hand, with a mouse, in tools that ask the user to
click along a curve or trace a point cloud \citep{webplotdigitizer}. This practice
has three deficiencies, and they motivate everything that follows.

\begin{description}
  \item[(D1) Low precision.] Hand-digitization recovers perhaps two or three
        significant figures, set by how steadily a human can place a cursor over a
        rasterized image. The recovered numbers inherit no stated error and no
        provenance.
  \item[(D2) Low automation, hence low throughput.] Because a person is in the
        loop for every point, the cost scales with the data. Digitizing one figure
        is tedious; digitizing the thousands of figures needed for a
        meta-analysis, an audit, or a corpus study is infeasible.
  \item[(D3) No verification.] Even a carefully digitized dataset cannot be proven
        to be the numbers the author actually plotted. Hand digitization offers no
        certificate of its own fidelity, and so no non-repudiable result that a
        third party is compelled to accept.
\end{description}

The key observation of this paper is that, for an important class of figures, none
of these deficiencies is fundamental. When a figure is stored as vector graphics
(the default for the \LaTeX-authored physics, astronomy, and computer-science
literature, and for many engineering datasheets), the document does not contain an
image of the plot. It contains the drawing instructions: the exact device
coordinates of every marker, the vertices of every polyline, the geometry of every
axis tick. The renderer wrote those coordinates with a fixed number of decimal places that,
for the dominant scientific toolchain, resolves more finely than the data's own
numerical type \citep{hunter2007matplotlib}. The data are therefore not approximated by the
figure; they are encoded in it, and can be decoded automatically, to high
precision, and with a certificate.

Extracting data from vector charts is not itself new
\citep{webplotdigitizer,chartdetective,svgdigitizer}, and the observation that
vector figures carry more numerical information than raster ones has been made
before, including in a notable manual forensic reanalysis of a single disputed
result \citep{hamlin2022}. Our contribution is not the possibility of extraction
but its formalization, automation, certification, and characterization. Concretely,
mirroring the three deficiencies above:

\begin{itemize}
  \item \textbf{A precision theory (addresses D1; \cref{sec:precision}).} We show
        that the achievable accuracy is governed by two separable steps, the
        renderer's coordinate quantization and the reader's axis calibration, and
        we bound each. For \texttt{matplotlib}'s PDF backend the quantization step
        returns bit-exact \float{32}; calibration is then the true bottleneck, at
        three to four significant figures, which already exceeds what most papers
        report for the same quantities.
  \item \textbf{An automatic extractor (addresses D2; \cref{sec:impl}).} A
        white-box parser paired with a black-box re-render decodes a figure in a
        single automatic pass, with no human in the loop, in place of the slow
        point-by-point tracing a digitizer requires; because nothing is interactive,
        the cost no longer grows with the number of points, and the same step runs
        over many figures as easily as over one.
  \item \textbf{A verification theory (addresses D3; \cref{sec:collision}).} The
        core is a proof that the extraction itself is correct. Recovery is injective
        except on a characterized, vanishingly small interval near zero, and
        accidental agreement between unrelated data is astronomically unlikely, so
        when the recovered values re-render to the same markers, lines, and
        ticks the figure draws they are certified to be the data that produced the
        figure: a non-repudiable result. That the same certificate also exposes a figure
        altered after rendering is a secondary benefit.
\end{itemize}

We support the theory with results that need no trust. With no ground truth used
during recovery, decoded figures agree with independent public archives
(Planck~2018 to $\sim\!10^{-9}$ and the Keeling \chem{CO_2} record to
$\sim\!5\times10^{-4}$, \cref{sec:results}), and we measure the achievable
precision across common renderers and their PDF, SVG, and EPS export formats
(\cref{sec:renderers}). As a concrete and deliberately safe application, we decode
the figure behind the Chinchilla scaling law and independently reproduce a
published correction to its confidence interval (\cref{sec:application}). We
release \tool, an open-source verifier. Because decoding figures at this precision
could power a range of data-integrity checks, from spotting fabricated or altered
points to detecting reuse across figures, we are explicit about scope: this paper
concerns only the extraction and its verification, and passes no judgement on the
integrity of any particular figure or paper. As detailed in \cref{sec:ethics},
recovered values are candidates for human review, never automated accusations.

\section{Related work}
\label{sec:related}

We position our work against four lines of prior art. In each case the prior
method is real and useful; we state what it does and then the specific advantage
our approach adds, indexed to the deficiencies of \cref{sec:intro}.

\paragraph{Interactive raster digitizers.}
Tools such as WebPlotDigitizer \citep{webplotdigitizer} let a user recover
approximate data by clicking on a rasterized plot. They are general (they work on
any image) and widely used. But they are semi-automatic (a human places points)
and pixel-limited (precision is set by image resolution and hand steadiness), and
they produce no certificate. Our advantage is fully automatic recovery at the
renderer's native precision, with a re-render certificate, directly answering
D1--D3, at the cost of requiring a vector source.

\paragraph{Vector-chart extraction.}
ChartDetective \citep{chartdetective} and svgdigitizer \citep{svgdigitizer}
recover data from vector charts (PDF or SVG), and are far more accurate than raster
digitizers because they read coordinates rather than pixels. They remain, however,
interactive (the user identifies series, axes, and mappings) and they treat
extraction as an engineering task rather than a quantity with a provable ceiling.
Our advantage is a precision theory that says how many bits are recoverable for a
given renderer and format (\cref{sec:precision,sec:renderers}), automation of
series and axis identification, and a certificate on the output. We also show that
the common PDF$\to$SVG conversion route used by such pipelines discards roughly ten
to fourteen bits relative to parsing the PDF directly (\cref{sec:renderers}).

\paragraph{Raster image forensics.}
A large literature detects manipulated bitmap figures, the duplicated or spliced
blots and micrographs typified by large-scale screening for inappropriate image
reuse \citep{bik2016}. This work is essential where the figure is a photograph, but
it operates on pixels and does not recover the plotted numbers, so it cannot check
whether a curve is consistent with a claim. Our advantage is that for vector plots
we recover the numerical data and bound the probability that an apparent match is
coincidental, giving a quantitative, data-level forensic primitive rather than a
pixel-level one.

\paragraph{Manual forensic reanalysis.}
Most directly related, \citet{hamlin2022} manually digitized a disputed
condensed-matter figure (via a PDF$\to$SVG centroid extraction) and used the
recovered digital precision to argue about data provenance, demonstrating, by hand
and on a single case, that vector figures carry forensically meaningful
information. That work validates its recovered numbers only indirectly, though, by
comparison with external ground truth and by concordance across a few related
figures and renderings; it offers no direct verification mechanism, no certificate
that an extraction is correct on its own terms. None of the prior art above does. A
central original contribution of this paper is exactly such a mechanism: a
re-rendering certificate that ties the recovered values to the markers, lines, and
ticks the figure itself draws and so proves an extraction faithful without recourse
to any outside reference.
Around it we add the formal recoverability theory, the no-collision guarantee, the
ground-truth validation against external archives, and the map across renderers and
formats that together turn a one-off manual analysis into a general, automatic
method.

In short, neither high-precision vector extraction nor the idea that vector figures
are forensically informative is new. What we add is the theory that bounds the
precision, the automation that scales it, the measurements that delimit where it
works, and, above all, the verification certificate that makes a recovered result
provable rather than merely plausible.

\section{Preliminaries}
\label{sec:prelim}

We first describe, in renderer-agnostic terms, the pipeline that turns numbers
into a vector figure, since it is this pipeline, not any one library, that
determines what can be recovered. \texttt{matplotlib} appears only as a concrete
instance.

\paragraph{From data to device coordinates.}
A plotting library places a data value on the page by an affine map. For a
single axis, a data value $v$ becomes a device coordinate
\begin{equation}
  u \;=\; a\,v + b,
  \label{eq:affine}
\end{equation}
where the scale $a$ and offset $b$ are fixed by the axis limits and the physical
size of the plotting area; a log axis applies \cref{eq:affine} to $\log v$. The
pair $(a,b)$ is shared by every datum drawn against that axis, and the axis ticks
are themselves drawn through the same map, a fact we will use to recover $(a,b)$
without any external information.

\paragraph{From device coordinates to stored text.}
The vector container (PDF, SVG, or EPS) is a program: it records drawing operators
and their operands as text. A point at device coordinate $u$ is therefore written
as a decimal numeral with a fixed number of fractional digits set by the backend.
Writing $d$ fractional digits snaps $u$ to a grid of spacing
\begin{equation}
  \Delta_{\mathrm s} \;=\; 10^{-d}
  \label{eq:quantum}
\end{equation}
in the coordinate unit; we call $\Delta_{\mathrm s}$ the storage quantum.
Crucially, the data value is not rounded to ``plot resolution''; it is rounded only
at this final printing step, which is typically far finer than the eye, and (as we
will see) often finer than the data's own numerical type.

\paragraph{Markers, polylines, and series.}
Two drawing primitives carry essentially all plotted data. A marker (a scatter
point) is a small glyph placed at a point; a polyline (a line plot) is a connected
sequence of vertices. Backends commonly emit the points of one series as an
absolute first placement followed by incremental offsets, so recovering positions
requires accumulating those increments per series, a detail that matters for
correctness (\cref{sec:impl}) but not for the precision argument.

\paragraph{The matplotlib instance.}
In \texttt{matplotlib}'s PDF backend, marker centres are placed through a
coordinate-system transform written with ten fractional digits
($d=10$, $\Delta_{\mathrm s}=10^{-10}$), while polyline vertices are written with
six ($d=6$, $\Delta_{\mathrm s}=10^{-6}$). These two numbers, together with the
affine map of \cref{eq:affine}, are all the precision theory of
\cref{sec:precision} needs. Other renderers and export formats choose different
$d$, and we tabulate the consequences in \cref{sec:renderers}.

\paragraph{What ``recovery'' means.}
Given the stored coordinates and an estimate $(\hat a,\hat b)$ of the affine map,
we invert \cref{eq:affine} to obtain $\hat v=(\tilde u-\hat b)/\hat a$. The rest of
the paper asks two questions about $\hat v$: how close it is to $v$
(\cref{sec:precision}), and when a match between recovered values can be trusted as
non-coincidental (\cref{sec:collision}).

\section{Extraction precision}
\label{sec:precision}

We now bound the accuracy of the recovered value $\hat v$. Two quantization steps
stack, and the coarser one dominates; the surprising empirical fact is which one
it is.

\subsection{The two-step ceiling}

We first fix the two steps precisely. A datum $v$ is placed at $u=av+b$
(\cref{eq:affine}) and stored by rounding to the grid \cref{eq:quantum}, so the
stored coordinate is $\tilde u=u+\eta$ with
$\lvert\eta\rvert\le\tfrac12\Delta_{\mathrm s}$. The affine map is not known exactly
but recovered from the axis ticks as $\hat a=a(1+\delta_a)$ and $\hat b=b+\Delta_b$,
with the calibration error collected into a single relative quantity $\varepsilon$
such that $\lvert\delta_a\rvert\le\varepsilon$ and
$\lvert\Delta_b\rvert\le\varepsilon\lvert b\rvert$. Recovery inverts the estimated
map, $\hat v=(\tilde u-\hat b)/\hat a$.

\begin{theorem}[Extraction-precision ceiling]
\label{thm:ceiling}
With the setup above, the recovered value obeys
\begin{equation}
  \lvert \hat v - v\rvert \;\le\;
  \underbrace{\frac{\Delta_{\mathrm s}}{2\lvert a\rvert}}_{\text{storage step}}
  \;+\;
  \underbrace{\varepsilon\,\bigl(\lvert v\rvert + \lvert b/a\rvert\bigr)}_{\text{calibration step}}
  \;+\; O(\varepsilon^2),
  \label{eq:ceiling}
\end{equation}
where the storage step is fixed by the renderer and format, and the calibration
step by what the reader can reconstruct of the axes.
\end{theorem}

We prove this in \cref{app:precision}. The content of \cref{thm:ceiling} is that the
achievable error is the sum of the two terms, dominated by whichever is larger, and
that the two have very different sizes in practice.

\subsection{A worked matplotlib example}

Consider a scatter plot drawn by \texttt{matplotlib}'s PDF backend on a panel a
few inches wide. The marker storage quantum is
$\Delta_{\mathrm s}=10^{-10}$ in the transformed unit, and the scale $a$ maps the
axis range onto that panel. Plugging real values into the storage term of
\cref{eq:ceiling}, the half-ulp of storage lands below the spacing of the
\float{32} grid over the entire plotted range: inverting the stored coordinate
returns the very same \float{32} number the library started from, bit for bit, all
$24$ bits of significand \citep{ieee754}, not an approximation of them. We confirm this directly in
\cref{sec:results}: against data we control, marker recovery is exact to the last
bit; against Planck's released spectrum it agrees to $\sim\!10^{-9}$, i.e.\ to the
\float{32} floor.

Polylines are coarser. With $\Delta_{\mathrm s}=10^{-6}$ the storage term caps a
polyline vertex at roughly six to seven significant figures, still far beyond what
any hand digitizer reaches, but no longer bit-exact.

\subsection{Calibration is the real bottleneck}

The second term of \cref{eq:ceiling} is the one that bites. In practice the reader
does not know $(a,b)$ exactly; they are recovered from the tick marks, whose device
positions are themselves only stored coordinates, and whose data values are read
from the tick labels, whose printed precision sets the true limit
(\cref{app:calibration}). This injects a relative error $\varepsilon$ of order
$10^{-3}$ to $10^{-4}$ of the axis range. Since
$\varepsilon \gg \Delta_{\mathrm s}/\lvert a\rvert$ for any real figure, the
calibration term dominates and

\begin{equation*}
  \lvert \hat v - v\rvert \;\approx\; \varepsilon\,(\lvert v\rvert+\lvert b/a\rvert)
  \;\sim\; 10^{-3}\text{--}10^{-4}\ \text{of the axis range,}
\end{equation*}

i.e.\ three to four significant figures end to end. This is the number a user of
the method actually gets. Two remarks make it meaningful. First, it is not a
limitation of the figure: the data are stored far more precisely, and better
calibration (more ticks, a known axis transform, or author-supplied limits) moves
the ceiling back toward the storage floor. Second, three to four significant
figures already exceeds what most papers report for the same quantities, so even
the calibration-limited result is often more precise than the text it accompanies.
We return to this in \cref{sec:results}, where the Keeling polyline, in the
calibration-limited regime, reproduces the reference record to
$\sim\!5\times10^{-4}$, exactly the regime \cref{eq:ceiling} predicts.

The storage term, finally, is renderer-specific: it is what separates a backend
that resolves the \float{32} grid from one that does not. We quantify that across
the ecosystem in \cref{sec:renderers}.

\section{The no-collision property}
\label{sec:collision}

This section is what lets the method's forensic conclusions be stated rigorously
rather than suggestively. Recovering a number accurately is not enough to use it as
evidence; two further things must hold. First, the recovery must be unambiguous: two
different data values must never end up stored as the same coordinate, or we could
not say which one the author plotted. Second, an agreement must be meaningful: if a
recovered dataset coincides with another figure or a published reference, that
coincidence has to be far too unlikely to be chance. We call the first property
injectivity and the second a coincidence bound, and together they are what we mean by
``no collision''. Without them a forensic match is only a plausible story; with them
it becomes evidence. The two subsections make each precise.

\subsection{Injective recovery, and the interval where it fails}

\begin{theorem}[Injectivity and the collision interval]
\label{thm:inj}
Two distinct data values are stored as distinct coordinates whenever their device
separation exceeds the storage quantum,
$\lvert a\rvert\,\lvert v_1-v_2\rvert > \Delta_{\mathrm s}$, and may collide only when
it does not. Consequently recovery
is injective on any data whose neighbouring values are separated by more than
$\Delta_{\mathrm s}/\lvert a\rvert$. If the data are \float{32} numbers, the
representable grid at magnitude $\lvert v\rvert$ has spacing
$\mathrm{ulp}(v)=2^{\lfloor\log_2\lvert v\rvert\rfloor-23}$, so individual values
are resolvable when $\lvert v\rvert > x^\dagger$, where
\begin{equation}
  x^\dagger \;\approx\; \frac{\Delta_{\mathrm s}}{\lvert a\rvert}\,2^{24}.
  \label{eq:xdagger}
\end{equation}
Below $x^\dagger$, a dyadic interval $[-x^\dagger,x^\dagger]$ around zero, the
floating-point grid is finer than the storage grid and neighbouring values
collide.
\end{theorem}

The proof and the exact dyadic constant are in \cref{app:collision}.
\Cref{fig:collision} demonstrates the injectivity condition directly on real
renders: as markers are packed below a device's storage quantum, distinct inputs
begin to land on one stored coordinate, and the onset is set by that device's
measured quantum.

\begin{figure}[t]
  \centering
  \includegraphics[width=\linewidth]{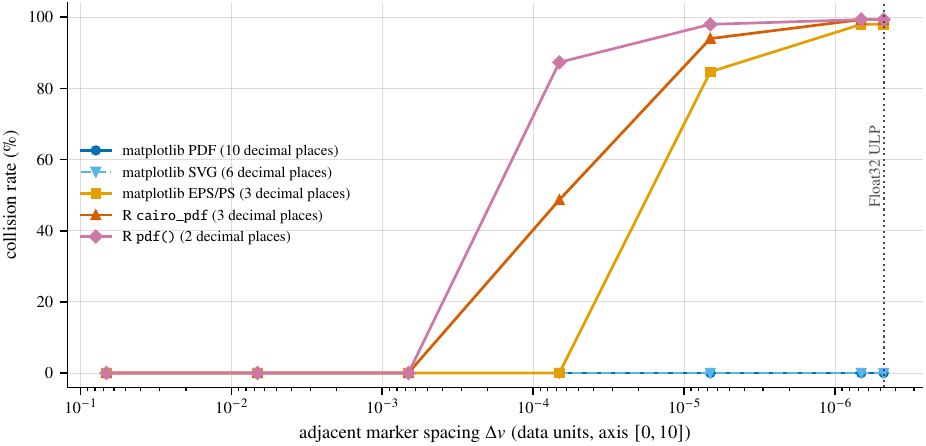}
  \caption{Injectivity, measured. We pack 150 distinct values into a shrinking
  window on a fixed axis, render them with each renderer, and re-extract; a
  collision is two distinct inputs that land on one stored device coordinate
  (\cref{thm:inj}). As the adjacent spacing $\Delta v$ falls below a renderer's
  storage quantum, collisions appear. matplotlib's PDF and SVG outputs (ten and six
  decimal places) resolve every marker down to the Float32 ULP with no collisions,
  and coincide along the bottom axis, whereas matplotlib's EPS/PS output (three
  decimals) and both R devices collapse, keeping only a handful of distinct values
  out of 150 at tight spacing. The same library loses most of its precision through
  the PostScript backend, so format matters as much as renderer. This is the
  renderer contrast of \cref{sec:renderers} at the level of individual points.}
  \label{fig:collision}
\end{figure}

\paragraph{The failure interval is rare, and rarity is renderer-specific.}
The width of $[-x^\dagger,x^\dagger]$ scales with the storage quantum, so the
renderer decides how often it matters. For matplotlib PDF, \cref{eq:xdagger} gives
$x^\dagger \approx 4\times10^{-6}$ of the axis range, so a datum must sit within
about $0.0004\%$ of zero to be at risk, which essentially never happens for data
plotted on a sensibly chosen axis. For PGFPlots the interval reaches \float{32}
only near the top of the scale. For R's PDF devices, by contrast, $x^\dagger$
exceeds the entire axis range, so these devices never resolve the \float{32} grid
and collisions are pervasive at that level, exactly as \cref{fig:collision} shows.
The point is that for the backend on which bit-exact recovery is claimed, the
collision region is a vanishing sliver, not a common occurrence.

\subsection{A match is evidence, not coincidence}

The second sense of ``no collision'' is statistical, and it is what gives the
method forensic force. Suppose two figures, or a figure and a reference dataset,
share $n$ recovered points. If those points are genuinely unrelated, how likely is
the agreement to be accidental?

\begin{theorem}[Coincidence bound]
\label{thm:coin}
Suppose each matched coordinate carries $b$ effective bits: after degenerate,
low-cardinality, and structured values are excluded, it is modelled, against an
adversary seeking a chance match, as no more concentrated than uniform over its
$2^{b}$ resolvable cells. Let two point sets of size $n$ be independent, with the two
coordinates of a point independent, and call two points matched when both
coordinates agree to within one cell. Then the probability that all $n$ points of
one set match the other is at most
\begin{equation}
  \Pr[\text{accidental match}] \;\le\; 2^{-2 n b}.
  \label{eq:coin}
\end{equation}
\end{theorem}

For modest values this is astronomically small: $n=20$ points at $b=10$ effective
bits gives $2^{-400}\approx 10^{-120}$. A multi-point agreement is therefore not
something that happens by chance; it is evidence that the two figures share a
common data origin, or that a recovered figure faithfully reproduces the
reference. We use \cref{eq:coin} in two ways: to certify a recovery (a decoded
figure whose re-render reproduces the original's markers, lines, and ticks agrees on every point, an event
of negligible coincidence probability) and to flag genuine data reuse across
figures.

\paragraph{Why all points, not one.}
The force of \cref{eq:coin} comes from requiring all $n$ points to agree at once, not
from any single coincidence. A lone match carries little weight: two unrelated series
land on the same value at one point with probability about $2^{-2b}$, which for the
$b=10$ above is roughly $10^{-6}$, small but not rare once a whole corpus is scanned.
Asking instead whether some one of the $n$ points happens to coincide is the wrong
test: that probability grows with $n$, to about $n\,2^{-2b}$, and is exactly the rate
of accidental single matches the method must rise above. Requiring simultaneous
agreement on the entire set behaves the opposite way: every further point multiplies
the improbability, so $n$ points give $2^{-2nb}$ and the evidence compounds. A single
coincidence is noise; agreement on all $n$ points together is the signal, which is why
both the certificate and the reuse test demand the whole set, never a part of it.

\paragraph{The entropy gate.}
Bound \cref{eq:coin} is only as strong as $b$. Integer-valued, low-cardinality, or
otherwise structured sequences carry few effective bits, and unrelated figures can
share them by chance; two plots both running $1,2,\dots,10$ on the $x$-axis is not
evidence of anything. We therefore gate the forensic test to high-entropy
sequences, estimating $b$ per series and excluding the degenerate cases before any
claim is made. This gate is what keeps a large bound from becoming a false
accusation, and it is applied uniformly in \cref{sec:results}.

\section{Implementation}
\label{sec:impl}

Our verifier, \tool, is a white-box parse paired with a black-box re-render. We
describe the steps that bear on precision and on the certificate, and summarise
axis calibration here while deferring its details to \cref{app:calibration}; the
lower-level page plumbing, walking the page tree and tracking the graphics state, is
standard and we pass over it.

\paragraph{Reading coordinates without losing bits.}
The decisive choice is to parse the vector container directly, never through a
conversion. We read the page content stream and follow the graphics state through
nested form XObjects, collecting the operands of the marker- and line-drawing
operators as written, at the backend's full $d$ fractional digits
(\cref{eq:quantum}). This is what preserves the storage floor of
\cref{sec:precision}: as we quantify in \cref{sec:renderers}, routing through
PDF$\to$SVG (a common shortcut) re-quantizes the coordinates and discards roughly
ten to fourteen bits, collapsing bit-exact recovery into mere digitization.

\paragraph{Reconstructing series.}
Because backends emit a series as an absolute first point followed by incremental
offsets, we segment the stream into per-series blocks and reconstruct positions by a
cumulative sum within each block. Getting this segmentation right is necessary for
correctness, since an off-by-one in the accumulation corrupts an entire series, but
it introduces no error of its own: the arithmetic is exact in the stored precision.

\paragraph{Calibrating the axes.}
The affine map is recovered from the ticks the renderer itself drew, with no axis
limits or metadata assumed. We segment the page into panels from the drawn frames,
collect the text spans below and to the left of each panel as candidate tick labels,
and parse them into values, recombining a base ``10'' with its raised exponent into
$10^{k}$, expanding SI suffixes such as ``1B'' or ``10T'', and reading scientific
notation directly. From the resulting (position, value) pairs we fit a linear and a
logarithmic map and keep the one with the smaller residual, which also serves as a
quality gate: a panel that does not fit cleanly is reported as uncalibrated rather
than guessed. \Cref{app:calibration} gives the full procedure and identifies
tick-label rounding as the dominant source of the calibration error $\varepsilon$ of
\cref{eq:ceiling}.

\paragraph{Inverting the quantization.}
With coordinates in hand and the affine map recovered, we invert \cref{eq:affine}.
When the backend resolves the \float{32} grid (\cref{sec:renderers}) and calibration
is tight, we additionally snap to the unique \float{32} pre-image, recovering the
exact value the author held in memory; when it is not, we report the digitized value
with the precision \cref{eq:ceiling} allows. Calibration, the practical bottleneck
(\cref{sec:precision,app:calibration}), contributes only the $\varepsilon$ term
already accounted for.

\paragraph{Certifying by re-rendering.}
Recovery alone is not a proof. We therefore re-render the recovered data with the
same library and compare the result against the original at the byte level, but only
over the elements that carry data: the marker placements, the polyline vertices, and
the axis-tick geometry, matched operator for operator. The comparison deliberately
excludes everything that encodes no plotted value, the text of axis labels, titles,
and annotations, the embedded fonts, and the document metadata, since two faithful
renders can differ in these without differing in the data, and folding them in would
only weaken the test. Agreement on the data-bearing operators is the certificate: by
the coincidence bound \cref{eq:coin} it cannot arise unless the recovered data are
the plotted data, so the result is non-repudiable. The same comparison makes
tampering evident: an edit to the plotted geometry after the original render cannot
reproduce those operators without the original data, closing the loop between what we
decoded and what the document contains. The one assumption is that the original's
renderer can be reproduced; matplotlib's content stream is stable across releases, so
in practice re-rendering matches without knowing the exact version.

\paragraph{Cost.}
Every step is linear in the number of drawn primitives and involves no human
interaction, so the verifier runs unattended over large collections; this is the
automation that answers deficiency~D2 of \cref{sec:intro}.

\section{Results}
\label{sec:results}

The strongest test of a decoder is to decode a published figure whose underlying
data exist in an independent public archive, without using that archive during
recovery, and then compare. We do this in the two regimes predicted by
\cref{sec:precision}: markers, the machine-precision regime, and polylines, the
calibration-limited regime. In both, the recovered data match the reference, and
the few residual discrepancies are explained by the theory rather than signalling
a failure of the method. Both flagship figures are matplotlib PDFs, the one backend
that resolves the \float{32} grid (\cref{sec:renderers}); on other renderers the
machine-precision regime is unavailable and recovery falls to the lower,
calibration-limited ceilings mapped there, so the bit-exact result below is
matplotlib-specific while the polyline result is representative of any vector
backend.

\paragraph{Markers: the machine-precision regime.}
The Planck~2018 temperature power spectrum \citep{planck2018} is published as both
a figure and a data release. Decoding the figure's markers and calibrating from
its ticks reproduces the released multipole and power pairs to $\sim\!10^{-9}$,
the \float{32} floor anticipated in \cref{sec:precision}. No Planck data entered
the pipeline. Against figures we generate ourselves, where the ground truth is
known exactly, marker recovery is bit-exact: the recovered \float{32} values equal
the inputs on every bit, consistent with \cref{thm:ceiling} and with the
injectivity of \cref{thm:inj}.

\paragraph{Polylines: the calibration-limited regime.}
The Keeling \chem{CO_2} curve, as plotted in a recent preprint, decodes to the
NOAA Mauna Loa monthly record \citep{noaa_co2} to $\sim\!5\times10^{-4}$ of range.
This is not a shortfall of the figure but the calibration term of
\cref{eq:ceiling} made visible: a smooth polyline offers fewer independent
constraints on the affine map than a field of discrete markers, so the calibration
error is larger. The residual is consistent with tick-calibration error and with
nothing in the stored curve.

\paragraph{Explaining the exceptions.}
Where individual points disagree with the reference, the cause is always one of the
mechanisms the theory names, never an unexplained failure: values inside the dyadic
interval $[-x^\dagger,x^\dagger]$ of \cref{thm:inj}, where neighbouring floats are
genuinely unresolvable; calibration error on figures with sparse or ambiguous
ticks, bounded by \cref{eq:ceiling}; and points occluded by overplotting, which the
parser flags rather than guesses. Each is predicted, each is detectable, and each
shrinks under better calibration, so the agreement with ground truth is not only
close but accountable. Beyond these two flagship cases, a larger audit of recovered figures against their
deposited source data, across several renderers, gives a median relative error of
about $0.09\%$, in the calibration-limited regime; the MATLAB and OriginLab rows of
\cref{tab:renderers} are two figures from this audit, while that table's top block
reports separate controlled round-trips that isolate each renderer's storage floor.

\section{Renderer and format coverage}
\label{sec:renderers}

The storage term of \cref{eq:ceiling} is a property of the exporter, not of the
plotting library's intent, so the right unit of analysis is the pair
(renderer $\times$ export format). \Cref{tab:renderers} characterizes coverage on a
single bit scale: for each pair it reports the storage precision and a measured
round-trip error.

We read the table two ways. The extractable precision
$B_{\text{extract}}=\log_2(\text{axis span}/\Delta_{\mathrm s})$ is the recoverable
significand at full scale, that is, how accurately data can be liberated. Forensics
needs more, in two grades. Recovering the exact stored \float{32} value requires only
$B_{\text{extract}}>24$, the significand depth, a bar SVG clears and PGFPlots reaches.
Provenance fingerprinting is stronger: it requires enough headroom above $24$ bits for
the recovered data to re-render to the original page bytes, and only matplotlib PDF,
at about $42$ bits, has it.

\begin{table}[t]
  \centering
  \small
  \setlength{\tabcolsep}{8.5pt}
  \caption{Precision by renderer and export format. $B_{\text{extract}}$ is the
  significand bits recoverable at full scale (defined above); ``Resolves
  \float{32}?'' asks whether that exceeds the $24$-bit \float{32} grid, the
  condition for recovering the exact stored value; the ``(provenance)'' tag marks the
  stronger case where the headroom above $24$ bits also fingerprints the render. The
  last column is a measured round-trip error: the median relative error between
  recovered and true values. Top block: a controlled round-trip (render the same
  known dataset, re-extract, fit the affine map, take the residual), which exposes
  each renderer's storage floor and confirms the $B_{\text{extract}}$ ceiling (the
  error tracks $2^{-B_{\text{extract}}}$; matplotlib PDF recovers \float{32}
  bit-for-bit). Bottom block: renderers we could not run locally, measured instead
  on a real published figure with deposited data, where calibration rather than
  storage sets the floor (\cref{sec:precision}) and the error is sub-percent.
  EPS/PS is Encapsulated PostScript; matplotlib writes only three marker decimals
  there, so its PostScript output falls below the \float{32} grid.}
  \label{tab:renderers}
  \begin{tabular}{llccl}
    \toprule
    Renderer & Format & $B_{\text{extract}}$ & Resolves \float{32}? & Round-trip error \\
    \midrule
    matplotlib            & PDF    & $\sim$42   & yes (provenance) & $3{\times}10^{-13}$ (bit-exact) \\
    matplotlib            & SVG    & $\sim$32   & yes              & $8{\times}10^{-10}$ \\
    matplotlib            & EPS/PS & $\sim$19   & no               & $9{\times}10^{-7}$ \\
    PGFPlots / TikZ       & PDF    & $\sim$25   & borderline       & $4{\times}10^{-8}$ \\
    R \texttt{cairo\_pdf} & PDF    & $\sim$18   & no               & $3{\times}10^{-6}$ \\
    R \texttt{pdf()}      & PDF    & $\sim$14.5 & no               & $7{\times}10^{-6}$ \\
    \midrule
    \multicolumn{5}{l}{Measured on a real published figure (calibration-limited; not run locally):} \\
    MATLAB     & \multicolumn{3}{c}{from a deposited-data figure} & $0.12\%$ \\
    OriginLab  & \multicolumn{3}{c}{from a deposited-data figure} & $0.25\%$ \\
    \bottomrule
  \end{tabular}
\end{table}

\paragraph{What the table says.}
Three points stand out. First, format matters as much as renderer: the same
matplotlib data recovers \float{32} bit-for-bit through the PDF backend (ten marker
decimals) and to nine digits through SVG (six decimals), but collapses to three or
four significant figures through the PostScript backend (three decimals), below the
\float{32} grid. Second, only matplotlib PDF pairs full \float{32} recovery with the
extra headroom (about $42$ bits) that fingerprints provenance; SVG and PGFPlots
reach \float{32}-grade data recovery but not provenance. Third, the controlled
round-trips track $2^{-B_{\text{extract}}}$ and so confirm the storage ceilings,
whereas real published figures round-trip only to sub-percent accuracy because there
calibration, not storage, sets the floor, exactly as \cref{sec:precision} predicts;
the median over our ground-truth audit is about $0.09\%$.

\paragraph{Beyond the table.}
Two pipeline facts sit alongside these renderers. Converting a PDF to SVG with
\texttt{pdf2svg} re-quantizes matplotlib's markers to about $28$ bits, which is the
route behind the manual reanalysis of \citet{hamlin2022} and why \tool{} parses PDF
directly (\cref{sec:impl}); and Plotly writes seven decimals but its Skia backend
caps effective precision near \float{32}. Our extractor also handles ROOT, gnuplot,
xmgrace, Mathematica, IgorPro, and Adobe Illustrator figures at the
digitization-to-\float{32} grades, validated by round-trip. For the low-precision
devices the collision interval $[-x^\dagger,x^\dagger]$ of \cref{thm:inj} covers
the whole axis, so they cannot fingerprint, but four to five significant figures is
still ample for data liberation and for the coincidence-based forensics of
\cref{thm:coin}. The honest boundary is that provenance is matplotlib-PDF
territory, while data liberation spans the whole ecosystem.

\section{Application}
\label{sec:application}

A safe way to show what the method is for is to revisit a question that is already
public, already in the literature, and entirely about a number rather than about
conduct: how wide the confidence interval on the Chinchilla compute-optimal exponent
should be. \citet{besiroglu2024chinchilla} reanalysed the scaling-law fit of
\citet{hoffmann2022chinchilla} and argued that the reported interval on the exponent
$a=\beta/(\alpha+\beta)$ was implausibly tight, a point the original authors
acknowledged. We take that published correction as our starting point and ask a
narrow methodological question: can it be confirmed directly from the figure, at the
renderer's precision, without re-running anyone's experiments?

It can. The per-run points behind the relevant isoFLOP figure were not released as a
table, so the figure is their only public record. We decode $115$ markers from the
vector figure, calibrate from its logarithmic parameter axis and linear loss axis
(\cref{app:calibration}), and refit the reported model
$L=E+A/N^{\alpha}+B/D^{\beta}$ (\cref{fig:chinchilla}a). Recovering all $115$ points
is a single automatic pass over the figure; tracing that many overlapping markers by
hand at comparable precision would be slow and would itself inject the digitization
noise our reading is meant to exclude. Bootstrapping on the
authors' own protocol gives a confidence interval on $a$ of width about $0.045$,
more than forty times the reported $\approx 0.001$ and in close agreement with the
reanalysis of \citet{besiroglu2024chinchilla} (\cref{fig:chinchilla}b). Because we
read the figure at its storage precision rather than by hand, our result removes
digitization noise as a possible objection to that reanalysis: the wide interval is
present in the authors' own plotted points, not an artefact of imprecise tracing.

Two points keep this in proportion. First, the scaling law's headline finding, that
parameters and training data should scale together with $a\approx 0.5$, is untouched:
our central estimate $a\approx 0.50$ agrees with it, and only the stated precision of
one exponent is at issue. Second, the narrow interval has an innocent and
reproducible explanation, which we record as an observation and not a charge: it
reappears when each bootstrap resample is warm-started from the single converged fit
instead of refit independently, a common shortcut. We present the case in this form,
as the independent confirmation of an existing, acknowledged correction, because that
is the method's intended use: liberating the data of record so that published
conclusions can be checked, combined, and, where warranted, refined, with every
recovered value offered as a candidate for review (\cref{sec:ethics}).

\begin{figure}[t]
  \centering
  \includegraphics[width=\linewidth]{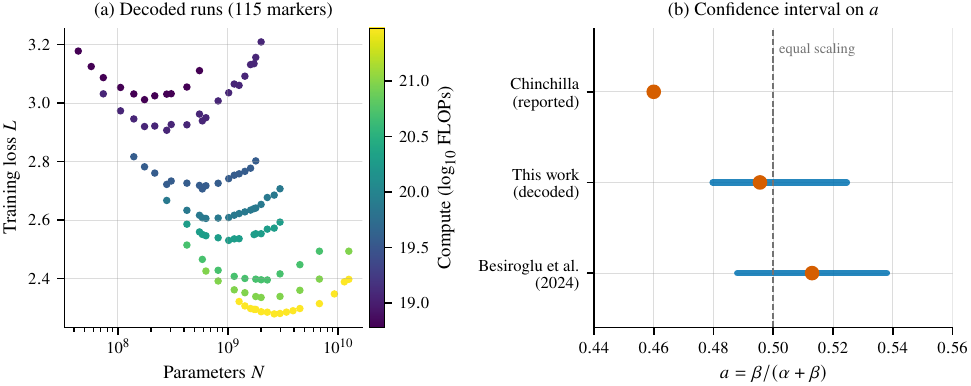}
  \caption{Confirming a published correction from the figure alone. (a) The training
  runs we recover from the vector figure of \citet{hoffmann2022chinchilla}: $115$
  markers, each a model's final training loss $L$ against its parameter count $N$,
  coloured by training compute in FLOPs. At a fixed compute budget the runs form one
  isoFLOP profile, whose envelope fixes the compute-optimal model size; we plot our
  recovered data, not the published image. (b) The exponent $a=\beta/(\alpha+\beta)$
  that sets how optimal model size scales with compute, where $a\approx 0.5$ means
  parameters and training data should grow together. Bars are confidence intervals
  and dots are point estimates. ``Chinchilla (reported)'' is the original parametric
  fit (their Approach~3); its interval (width $\approx 0.001$) is a near-point,
  whereas our independent-refit bootstrap of the decoded data and the reanalysis of
  \citet{besiroglu2024chinchilla} both give intervals more than forty times wider and
  mutually consistent around $a\approx 0.5$. Our role is to reproduce that published
  correction directly from the figure.}
  \label{fig:chinchilla}
\end{figure}

\section{Responsible use}
\label{sec:ethics}

Because this work recovers data that authors did not explicitly release, and
because its forensic mode can surface apparent inconsistencies, we adopt explicit
safeguards.

\begin{itemize}
  \item \textbf{Candidates, not verdicts.} Every output, whether a recovered
        dataset, a cross-figure match, or a failed re-render certificate, is a
        candidate for human review, accompanied by its uncertainty
        (\cref{sec:precision}) and its coincidence bound (\cref{sec:collision}).
        The system makes no accusation, and we name no individual or paper as
        fraudulent. The one named figure we analyse in depth
        (\cref{sec:application}) concerns a published, already-adjudicated
        numerical dispute, not conduct.
  \item \textbf{Read-only and reproducible.} All external access is read-only; we
        take no outward or irreversible action. The verifier is deterministic, so
        any reported result can be independently re-derived from the same input.
  \item \textbf{Provenance, not surveillance.} We decline privacy-sensitive uses.
        Where a capability could be misused, for instance de-anonymizing an author
        through render fingerprints, we evaluate it only on synthetic data and
        report it as a measured limit, not as a tool.
  \item \textbf{Corrections through normal channels.} A recovered correction
        (\cref{sec:application}) is a scientific contribution to be raised through
        comment, contact with authors, or reanalysis, not a public allegation.
\end{itemize}

\paragraph{Distribution and copyright.}
We release the verifier and a metadata-only resource index. We do not redistribute
any third-party figure or PDF; published figures remain the property of their
authors and publishers, and our artifacts reference them by identifier rather than
reproducing them. Self-generated figures used for ground-truth tests are released
in full, since we own them.

\section{Conclusion}
\label{sec:conclusion}

Data-bearing vector figures are a high-bit container that the literature has been
treating as a picture. We gave a precision theory that bounds how much of their
contents can be recovered (bit-exact \float{32} for \texttt{matplotlib} markers,
calibration-limited three-to-four significant figures end-to-end), and a
no-collision theory under which recovery is injective away from a characterized,
vanishing interval near zero and an apparent match is evidence rather than
coincidence. An automatic verifier realizes both, recovering data with no human in
the loop and certifying the result by re-rendering its markers, lines, and ticks to the original. With no
ground truth in the loop the decoder reproduces external archives to
$\sim\!10^{-9}$ (markers) and $\sim\!5\times10^{-4}$ (polylines), and on a single
famous figure it independently reproduces a published correction to a scaling-law
confidence interval. A map across renderers and export formats delimits where each
guarantee holds.

Together these turn the three deficiencies of hand-digitization (low precision,
low automation, no verification) into provable, automatic, certifiable
properties, for the class of figures that carry their data in vector form. Our aim is to make figure-encoded data a first-class, checkable part of the
scientific and technical record.

\section*{Code and data availability}
The verifier \tool, together with its unit tests and self-contained fixtures,
the renderer-precision measurements, and a metadata-only index of the figures
analysed, is available under the MIT license at
\url{https://github.com/Bowen-Sun-0728/figverify}. It installs as a command-line
tool (\texttt{classify}, \texttt{extract}, \texttt{verify}) and reproduces the
experiments reported here. No third-party figures or PDFs are redistributed
(\cref{sec:ethics}).

\bibliographystyle{unsrtnat}
\bibliography{refs}

\appendix
\section{Proof of the extraction-precision ceiling}
\label{app:precision}

We prove \cref{thm:ceiling}, in the setup fixed there: the stored coordinate is
$\tilde u=u+\eta$ with $\lvert\eta\rvert\le\tfrac12\Delta_{\mathrm s}$ on the grid
\cref{eq:quantum}; the map recovered from the ticks is $\hat a=a(1+\delta_a)$ and
$\hat b=b+\Delta_b$ with $\lvert\delta_a\rvert\le\varepsilon$ and
$\lvert\Delta_b\rvert\le\varepsilon\lvert b\rvert$; and recovery inverts it as
$\hat v=(\tilde u-\hat b)/\hat a$.

\begin{proof}
Substitute $\tilde u = av+b+\eta$ and the calibration estimates:
\[
  \hat v
  = \frac{av+b+\eta-(b+\Delta_b)}{a(1+\delta_a)}
  = \frac{a v+\eta-\Delta_b}{a(1+\delta_a)}
  = \Bigl(v+\tfrac{\eta}{a}-\tfrac{\Delta_b}{a}\Bigr)\bigl(1+\delta_a\bigr)^{-1}.
\]
Expanding $(1+\delta_a)^{-1}=1-\delta_a+O(\delta_a^2)$ and discarding products of
two small quantities ($\eta\delta_a$ and $\Delta_b\delta_a$ are $O(\varepsilon^2)$
because $\eta$ is already at the storage scale),
\[
  \hat v - v
  = -v\,\delta_a + \frac{\eta}{a} - \frac{\Delta_b}{a} + O(\varepsilon^2).
\]
Taking absolute values and using the triangle inequality with the bounds on
$\eta$, $\delta_a$, and $\Delta_b$,
\[
  \lvert \hat v - v\rvert
  \le \frac{\lvert\eta\rvert}{\lvert a\rvert}
     + \lvert v\rvert\,\lvert\delta_a\rvert
     + \frac{\lvert\Delta_b\rvert}{\lvert a\rvert}
  \le \frac{\Delta_{\mathrm s}}{2\lvert a\rvert}
     + \varepsilon\lvert v\rvert
     + \varepsilon\Bigl\lvert\frac{b}{a}\Bigr\rvert
     + O(\varepsilon^2),
\]
which is the claim. \end{proof}

\begin{corollary}[Bit-exact recovery]
\label{cor:bitexact}
Suppose calibration is exact ($\varepsilon=0$) and the data are \float{32} numbers.
If the storage step is below half the floating-point ULP across the plotted range,
$\tfrac12\Delta_{\mathrm s}/\lvert a\rvert < \tfrac12\,\mathrm{ulp}(v)$ for all
plotted $v$, then $\hat v=v$ bit-for-bit after rounding to \float{32}.
\end{corollary}

\begin{proof}
By \cref{thm:ceiling} with $\varepsilon=0$, $\lvert\hat v-v\rvert\le
\Delta_{\mathrm s}/(2\lvert a\rvert) < \tfrac12\,\mathrm{ulp}(v)$, so $v$ is the
unique \float{32} number within half an ULP of $\hat v$; rounding $\hat v$ to
\float{32} returns $v$. \end{proof}

For \texttt{matplotlib}'s PDF backend, $\Delta_{\mathrm s}=10^{-10}$ in the
transformed unit gives $B_{\text{extract}}=\log_2(\text{span}/\Delta_{\mathrm s})
\approx 42$ bits, comfortably above the $24$-bit \float{32} significand everywhere
except the vanishing near-zero interval of \cref{thm:inj}, so the hypothesis of
\cref{cor:bitexact} holds away from that interval and marker recovery is bit-exact
when calibration is tight (\cref{sec:results}).

\section{Proofs of the no-collision properties}
\label{app:collision}

\subsection{Injectivity and the collision interval}

\begin{proof}[Proof of \cref{thm:inj}]
Two data values $v_1\ne v_2$ are stored as $\tilde u_i=\Delta_{\mathrm s}
\operatorname{round}(u_i/\Delta_{\mathrm s})$ with $u_i=av_i+b$. If
$\lvert u_1-u_2\rvert=\lvert a\rvert\,\lvert v_1-v_2\rvert>\Delta_{\mathrm s}$, then
$u_1$ and $u_2$ cannot round to the same grid value, so $\tilde u_1\ne\tilde u_2$
and recovery distinguishes them; conversely if the device separation is at most
$\Delta_{\mathrm s}$ they may share a grid value. Hence recovery is injective on
data whose consecutive separations exceed $\Delta_{\mathrm s}/\lvert a\rvert$.

Now let the data be \float{32} numbers. The representable grid near magnitude
$\lvert v\rvert$ has spacing $\mathrm{ulp}(v)=2^{\lfloor\log_2\lvert v\rvert\rfloor-23}$.
Neighbouring representable values are individually resolvable iff their device
separation exceeds the storage quantum,
\[
  \lvert a\rvert\,\mathrm{ulp}(v) > \Delta_{\mathrm s}
  \iff
  2^{\lfloor\log_2\lvert v\rvert\rfloor-23} > \frac{\Delta_{\mathrm s}}{\lvert a\rvert}
  \;=:\;\delta .
\]
Writing $\lfloor\log_2\lvert v\rvert\rfloor > 23+\log_2\delta$ and using
$\lvert v\rvert\ge 2^{\lfloor\log_2\lvert v\rvert\rfloor}$ gives the threshold
\[
  \lvert v\rvert > x^\dagger,
  \qquad
  x^\dagger = 2^{\lfloor\log_2(2\delta)\rfloor + 24}
  \;\approx\; \delta\,2^{24}
  = \frac{\Delta_{\mathrm s}}{\lvert a\rvert}\,2^{24},
\]
where the factor $2$ inside the floor accounts for round-to-nearest. Thus the
non-resolvable set is the dyadic interval $[-x^\dagger,x^\dagger]$, establishing
\cref{eq:xdagger}. \end{proof}

\paragraph{Numerical size.}
With $\delta=\Delta_{\mathrm s}/\lvert a\rvert$ expressed as a fraction of the axis
range, $x^\dagger\approx\delta\,2^{24}$. For \texttt{matplotlib} PDF, $\delta$ is
small enough that $x^\dagger\approx 4\times10^{-6}$ of the range; for R's PDF
devices $\delta$ is large enough that $x^\dagger$ exceeds the range, so no datum is
resolvable at the \float{32} level, consistent with \cref{tab:renderers}.

\subsection{The coincidence bound}

\begin{proof}[Proof of \cref{thm:coin}]
Model each effective coordinate, after the entropy gate of
\cref{sec:collision} has removed degenerate, low-cardinality, or structured
values, as a draw from a distribution that is, to an adversary trying to produce a
chance match, no more concentrated than uniform over its $2^{b}$ resolvable cells.
For two independent draws, the probability of landing in the same cell (within the
matching tolerance of one cell) is at most $2^{-b}$. A point has two coordinates,
assumed independent, so a single point matches with probability at most
$2^{-2b}$. The $n$ points are independent, so
\[
  \Pr[\text{all } n \text{ points match}] \;\le\; \bigl(2^{-2b}\bigr)^{n}
  \;=\; 2^{-2nb}.
\]
\end{proof}

\paragraph{Remarks.}
(i) The bound rests on the entropy gate, which fixes $b$ as a coordinate's honest
effective entropy by excluding degenerate, low-cardinality, and structured sequences.
What remains places at most $2^{-b}$ of its mass on any one resolvable cell
(min-entropy at least $b$), and that is what caps the chance-collision probability at
$2^{-b}$. Concentration the gate fails to remove would raise the chance of an
accidental match, not lower it, so the gate is deliberately conservative; genuinely
shared data match regardless. (ii) The exponent scales linearly in both the
number of points $n$ and the per-coordinate bits $b$, so even modest figures give
overwhelming evidence: $n=20$, $b=10$ yields $2^{-400}$. (iii) For a re-render
certificate (\cref{sec:impl}) the ``match'' is over every data-bearing primitive
simultaneously, every marker, vertex, and tick, so $n$ is the full point count and
the bound is far below any machine epsilon, which is why agreement on those
primitives is treated as proof of faithful recovery. (iv) The exponent counts two
independent coordinates per point, $2b$ bits; for a functional relation $y=f(x)$ a
point carries closer to $b$, so the honest exponent is $nb$, still far below machine
epsilon for any non-trivial figure.

\section{Recovering the affine map from ticks}
\label{app:calibration}

The calibration step of \cref{thm:ceiling} recovers the affine map $(a,b)$ of
\cref{eq:affine} using only the axis ticks the renderer itself drew, with no axis
limits, data range, or metadata assumed. Because this step, not the storage
quantum, sets the precision actually achieved (\cref{sec:precision}), we give it in
full. It has four stages.

\paragraph{Panels.}
A single content stream may hold several subplots, so we first segment the page
into panels. We take the rectangles the renderer drew for the axis frames, keep
those whose area is a sizable but not page-filling fraction of the page, and
de-overlap them greedily, rejecting a candidate that overlaps an accepted panel by
more than a set fraction. Each marker is later assigned to the panel whose frame
contains it, and every panel is calibrated on its own.

\paragraph{Tick labels.}
For each panel we collect the text spans the renderer placed just below the bottom
edge as candidate $x$ labels and just left of the frame as candidate $y$ labels,
each tagged with the device coordinate of its bounding box. Turning a label into a
value is where a naive parser fails, and three cases recur in scientific figures. A
logarithmic label is drawn as a base ``10'' with a raised, smaller exponent; we
recombine the pair into $10^{k}$, but only when the exponent span is set smaller and
lifted above the baseline, so that two ordinary neighbouring ticks such as ``10''
and ``15'' are never fused into a spurious $10^{15}$. Engineering labels such as
``100M'', ``1B'', or ``10T'' are expanded through their SI suffixes, and scientific
text such as ``6e18'' is read directly. We normalise the Unicode minus sign, the
multiplication sign, and thousands separators before parsing.

\paragraph{Linear or logarithmic fit.}
From the cleaned (position, value) pairs of one axis we fit two models by least
squares, a linear map $v=p+qu$ and a logarithmic map $\log_{10}v=p+qu$ (the scale and
offset $(a,b)$ of \cref{eq:affine} follow by inversion), and keep whichever has the
smaller residual, measured in device points. The residual doubles
as a quality gate: a panel whose best fit exceeds a few points, or that yields fewer
than two parseable ticks on an axis, is reported as uncalibrated rather than
guessed. The chosen map is inverted to decode each marker (\cref{sec:impl}); a
logarithmic axis decodes as $10^{p+qu}$.

\paragraph{Where the calibration error comes from.}
With the fit in hand, the relative error $\varepsilon$ of \cref{eq:ceiling} is set
not by geometry but by the tick labels. A printed label is a rounded display of the
true tick value: an axis annotated ``0.5'' fixes that tick only to the precision the
author chose to print, and the affine map inherits that uncertainty. The device
positions of the ticks are themselves stored coordinates, quantised at
$\Delta_{\mathrm s}$, which is negligible beside label rounding, while the fit
residual across several ticks is the third and usually smallest contribution. This
is why \cref{sec:precision} finds calibration, not storage, to be the bottleneck,
and why anything that adds independent constraints, more ticks, a declared axis
transform, or author-supplied limits, drives the achievable precision back toward
the storage floor. When the plotted data are \float{32} and the axis carries enough
ticks, $\varepsilon$ falls below the floating-point ULP and recovery is bit-exact
(\cref{cor:bitexact}).

\end{document}